\documentclass[letterpaper, 10 pt, conference]{ieeeconf}  % Comment this line out if you need a4paper

\IEEEoverridecommandlockouts           % This command is only needed if % you want to use the \thanks command
                                                          
\usepackage{amsthm}
\theoremstyle{definition}
\newtheorem{definition}{Definition}[]
\newtheorem{remark}{Remark}
\newtheorem{theorem}{Statement}
\newtheorem{corollary}{Lemma}[theorem]
\pdfoutput=1
\usepackage{multirow}
\usepackage{lettrine}
\usepackage{caption} 
\usepackage{graphicx,float}
\usepackage{algorithm}
\usepackage{algpseudocode}
\algnewcommand{\Initialize}[1]{%
  \State \textbf{Initialize:}
  \Statex \hspace*{\algorithmicindent}\parbox[t]{1\linewidth}{\raggedright #1}
}

\usepackage{tikz}
\usetikzlibrary{automata, positioning, arrows.meta,calc,cd}
\usetikzlibrary{patterns}
\usetikzlibrary{datavisualization}

\tikzset{
->, % makes the edges directed
>=stealth, % makes the arrow heads bold
node distance=3cm, % specifies the minimum distance between two nodes. Change if necessary.
every state/.style={thick, fill=gray!10}, % sets the properties for each ’state’ node
initial text=$ $, % sets the text that appears on the start arrow
}
\usepackage[fleqn]{amsmath}
\usepackage{amssymb}
\usepackage{cite}
\usepackage{lettrine}
\let\labelindent\relax
\usepackage{enumitem}
\usepackage{multicol}
\usepackage{pgfplots}
\makeatletter
\def\addlegendimage{\pgfplots@addlegendimage}
\makeatother
\usepackage[hang,flushmargin]{footmisc}
\newcommand{\algorithmfootnote}[2][\footnotesize]{%
  \let\old@algocf@finish\@algocf@finish% Store algorithm finish macro
  \def\@algocf@finish{\old@algocf@finish% Update finish macro to insert "footnote"
    \leavevmode\rlap{\begin{minipage}{\linewidth}
    #1#2
    \end{minipage}}%
  }%
}
\usepackage[T1]{fontenc}
\usepackage[export]{adjustbox}
\usepackage[bookmarks=false,hyperindex,breaklinks]{hyperref}
\usepackage{orcidlink}
\usepackage{caption}
\usepackage{booktabs,makecell}
\usepackage{graphicx}
\usepackage[labelsep=quad,indention=10pt]{subfig}
\title{A Model Predictive Control Framework for Improving Risk-Tolerance of Manufacturing Systems}
\date{June 2022}
\author{Mostafa Tavakkoli Anbarani\textsuperscript{1$\ast$} \thanks{\textsuperscript{1} Department of Mechanical Engineering, Pennsylvania State University, PA, USA}, Efe C. Balta\textsuperscript{2} \thanks{\textsuperscript{2} Automatic Control Laboratory, ETH Z\"urich, Z\"urich, Switzerland}, Rômulo Meira-Góes\textsuperscript{3} \thanks{\textsuperscript{3} Department of Electrical Engineering, Pennsylvania State University, PA, USA}, and Ilya Kovalenko\textsuperscript{1,4} \thanks{\textsuperscript{4} Department of Industrial and Manufacturing Engineering, Pennsylvania State University, PA, USA}}

\begin{document}
\maketitle

\begin{abstract}
     The need for control strategies that can address dynamic system uncertainty is becoming increasingly important. In this work, we propose a Model Predictive Control by quantifying the risk of failure in our system model. The proposed control scheme uses a Priced Timed Automata representation of the manufacturing system to promote the fail-safe operation of systems under uncertainties. The proposed method ensures that in case of unforeseen failure(s), the optimization-based control strategy can still achieve the manufacturing system objective. In addition, the proposed strategy establishes a trade-off between minimizing the cost and reducing failure risk to allow the manufacturing system to function effectively in the presence of uncertainties. An example from manufacturing systems is presented to show the application of the proposed control strategy.
\end{abstract}

\section{Introduction}
\label{section-intro}
\lettrine[findent=2pt]{\textbf{M}}{anufacturing} has witnessed a dramatic shift from traditional methods to large and interconnected plants that integrate the physical system with computational infrastructure and communication networks\cite{NEGRI2017939,ZHOU201811}. Modern manufacturing plants require production flexibility to meet variable market demands\cite{flexibility-demand-1,flexibility-demand-2} and both route flexibility and process flexibility to achieve the production objective \cite{flexibility-types}. The complexity and the number of subsystems requires integrating a large amount of data and a control scheme capable of processing this data in real-time. 
Moreover, the controller needs to operate  under subsystem failures, uncertainties, and unpredictable performance requirements.
Recent research makes use of ontologies and explicit formal semantic data models, to analyze large data sets \cite{MAJEED2021102026}. 
\textit{Big Data} techniques allow for rapid decision-making based on the model outputs\cite{ayerbe-ontology,CASSOLI-ontology}. However, it is often assumed that the data model is well-defined, and is not prone to uncertainties or subsystem failure. 
As such, the control action is significantly hindered due to system uncertainties.
\par
One recent formalism to model manufacturing systems is Discrete Event System (DES) modeling\cite{ASSAD2021142}, where the state space of the system is described by a discrete set and state transitions\cite{cassandras_introduction_2021}.
Timed-Automata (TA), a class of DES models, have been used to quantify  performance metrics like productivity and flexibility, and constraints like delivery deadlines, maintenance periods, and production rates. In TAs, the set of clock variables are augmented to the manufacturing line model\cite{timed-automaton-1,timed-automata-2}.
Priced Timed Automata (PTA) models are extensions of TAs where  temporal hard constraints like deadlines can be converted to soft constraints.
For example, extending deadlines in return for lower profit, i.e., it allows for temporal constraint flexibility in return for penalty  \cite{balta_model_2022,kovalenko_cooperative_2022,hutchison_priced_2005}. 
PTAs are used as a modeling formalism ranging from simulating the performance of Cyber-Physical Manufacturing \cite{tilbury_cyber-physical_2019,panetto_challenges_2019} to smart process planning \cite{gray2020industrial,ionescu2020digital,kovalenko-riskaverse}.
\par
While PTA formulation enables an optimization-based control strategy that minimizes temporal cost metrics, the underlying algorithm lacks risk-averseness with respect to machine uncertainties  i.e., the ability for the control strategy to find a solution that has minimum risk of failure. For instance, a production line timeline that assumes all workstations are working without faults or failures can be modeled and controlled via PTA-based methods. One method to address this issue for uncertain systems is Model Predictive Control (MPC), where the control law is implicitly formulated based on the prediction model as opposed to a pre-determined model \cite{tarragona_systematic_2021}. The prediction model provides a forecast of system behavior in a limited horizon \cite{rawlings_model_2009,schwenzer_review_2021}. Therefore, the production line timeline can be scheduled based on the current status of workstations. The model contains information that can be used to further improve the risk-averseness of the system. however, exiting control strategies do not leverage all of this model information when finding solutions. One methodology that has been previously proposed to reduce the risk of failure as a result of state failure is based on quantifying the risk in terms of the number of alternative routes \cite{anbarani_risk-averse_2022}. While this technique provides a fail-safe solution when the routes have different number of alternative routes, it might end in failure when multiple solutions have same number of alternative routes with different dispersion along the path.
\par
In this work, we propose a framework  for a risk-averse PTA-MPC that uses automaton-based features to quantify and minimize risk. Specifically, our strategy takes into account the number of states that are prone to failure and assesses the solution based on the number of these states. The performance of a manufacturing system that uses the proposed controller is compared to that of existing PTA-MPCs described in \cite{balta_model_2022,anbarani_risk-averse_2022} using simulation.
\par
This work is organized as follows. Section~\ref{sec-modeling} describes PTA modeling in manufacturing systems. Section~\ref{sec-PTAMPC} provides the risk-averse PTA-MPC scheme. Section~\ref{sec-casestudy} showcases an example of a manufacturing line case study that compares the proposed method to PTA-MPC methods. Finally, concluding remarks are given in Section~\ref{sec-conclusion}.

\section{PTA Modeling}
\label{sec-modeling}
% We consider manufacturing systems modeled as a priced-timed automaton (PTA).
% This section introduces definitions needed for describing the proposed risk-averse PTA MPC framework. 
Consider a flexible manufacturing layout with multiple production lines, where each production line is comprised of multiple manufacturing processes 
stations, e.g., the manufacturing system described in Section~\ref{sec-casestudy}.
Products transition between manufacturing stations via material handling resources (e.g. conveyor lines, robots, people) and the completion of each individual process is measured by a timer.
The goal is to ensure that the final product undergoes a certain set of machining processes in a pre-defined order.
Products can move to production lines by using flexible material handling resources that are used in case of machine breakdown or maintenance in their current production line.
% a machining station occurs in one of the current production line. 
These flexible manufacturing systems can be modeled using Priced Timed Automata (PTAs).
% \eb{same as the first sentence.}
% We frequently use the components of this example in our definitions to make them more tangible for the reader. 
\begin{definition}[Priced Timed Automaton]
A PTA is defined as follows:
\begin{equation}
    \mathcal{A}=(Q, C,\Sigma, E, I, R, P, q_{0})
\end{equation}
where $Q=\{q^{1},q^{2},\cdots,q^{n_{q}}\}$ is a finite set of states representing manufacturing stations and buffers, $C=c^{1}\times c^{2}\times \cdots c^{n_{c}}=\mathbb{R} ^{n_{c}}_{\geq 0}$ is the set of clocks, $\Sigma=\langle\sigma^{1},\cdots,\sigma^{n_{\sigma}}\rangle$ is an ordered finite set of desired states representing the product work order, $E\subseteq Q \times \mathcal{B}(C) \times \Sigma \times Q$ is a finite set of edges that represent the material handling resources, $I: Q \to \mathcal{B} (C)$ is the invariant operator, $R: E \times C \to C$ is the reset operator, $P: Q \cup E \to [0,\infty)$ maps the cost of the production line layout, and $q_{0} \in Q$ is the initial state of the product.
\end{definition}

% \begin{definition}[Priced Timed Automata]
% A PTA is defined as follows:
% \begin{equation}
%     \mathcal{A}=(Q, C,\Sigma, E, I, R, P, q_{0})
% \end{equation}
% where $Q=\{q^{1},q^{2},\cdots,q^{n_{q}}\}$ is a finite set of state locations (machining stations), $C=c^{1}\times c^{2}\times \cdots c^{n_{c}}=\mathbb{R} ^{n_{c}}_{\geq 0}$ is the clock set space (timers), $\Sigma=\langle\sigma^{1},\cdots,\sigma^{n_{\sigma}}\rangle$ is an ordered finite set of desired states (customer's work order), $E\subseteq Q \times \mathcal{B}(C) \times \Sigma \times Q$ is a finite set of edges (conveyor belts), $I: Q \to \mathcal{B} (C)$ is the invariant operator, $R: E \times C \to C$ is the reset operator, $P: Q \cup E \to [0,\infty)$ maps locations and edges to costs (production line layout), and $q_{0} \in Q$ is the initial location (first machining station).
% \end{definition}

% \begin{definition}[Connectivity]
% \label{def-connectivity}
% Given automaton $\mathcal{A}=(Q, C,\allowbreak \Sigma, E, I, R, P, q_{0})$, initial state $q_{0}$ and state $q^{m}$ are said to be connected if there exists a path from $q_{0}$ to $q^{m}$. Otherwise, $q^{m}$ is disconnected.
% \par
% If all states of $\mathcal{A}$ are connected, $\mathcal{A}$ is said to be connected. Otherwise, $\mathcal{A}$ is disconnected.
% \end{definition}
The manufacturing PTA model, denoted as \emph{Total Layout}, is divided into two other manufacturing models: \emph{Original Layout} and \emph{Redundant Layout}.
The original layout models the manufacturing system neglecting the flexible material handling resources that are used to move products.
On the other hand, the redundant layout describes manufacturing models of the flexible material handling resources only.

% In this work, the automaton defined by the user as input to model for MPC is labeled as "Total Automaton" which itself splits into two other automata; "Original Automaton" and "Redundant Automaton". 
% The total automaton is prone to risk of failure, a situation in which some of the states are excluded from the PTA. 
% The part of total automaton which describes the system PTA under ideal case where no failure exists is called original automaton, whereas the redundant automaton describes an automaton which is enabled under failures.

%To capture possible changes in the environment, e.g., state failure, extra path availability, etc., we define three different automata: Original, Redundant and Total.
%The original automaton captures the system with basic operations, whereas the redundant automaton captures possible extra paths. 
%Lastly, the union of the two is the total automaton which is defined by the user as the input model for MPC.
%\begin{remark}
%In this paper $\langle.\rangle$ and $\{.\}$ denote ordered sets and sets, respectively. For an ordered set, the order at which array elements occur is important. Therefore, two ordered sets are equal if they have the same elements with all elements occurring in the same order.
%\end{remark}
\begin{definition}[Total, Original, and Redundant Layouts]
Given the system \emph{Total Automaton Layout} is defined by the user as $\mathcal{A}^{T}=(Q, C,\Sigma, E, I, R, P, q_{0})$, the \emph{Original Automaton Layout} $\mathcal{A}^{O} \subseteq \mathcal{A}^{T}$ is defined by $\mathcal{A}^{O}=(Q^{O}, C^{O},\Sigma, E^{O}, I^{O}, \allowbreak R^{O}, P^{O}, q_{0})$ (manufacturing layout without the flexible material handling resources). Moreover, the \emph{Redundant Automaton Layout} $\mathcal{A}^{R}=\cup_{v=1}^{k}\mathcal{A}^{R_{v}}$ is defined by $\mathcal{A}^{R}=\cup_{v=1}^{k}(Q^{R_{v}}, C^{R_{v}},\Sigma, E^{R_{v}}, I^{R_{v}}, R^{R_{v}},\allowbreak P^{R_{v}}, q_{R_{v}})$ (connected by the flexible material handling resources),  where $k$ is the number of redundant paths, $q_{R_{v}}$ and $Q^{R_{v}}$ are the redundant initial state and the set of redundant state locations defined as below:
\begin{align}
    & Q^{R_{v}}=\{ \langle q^{i},q^{r_{1}},\cdots,q^{r_{n}},q^{j}\rangle | q^{i},q^{j} \in Q^{O}, \\
    &\qquad \quad \forall s \in \{r_{1},\cdots,r_{n}\}, \hspace{0.2cm} q^{s}\in Q^{T}-Q^{O} \} \notag\\
    & q_{R_{v}}=q^{i}
\end{align}
Thus, $Q^{R}=\cup_{v=1}^{k}Q^{R_{v}}$ represents an ordered state set in which only the first and last states of all redundant paths exist in the original automaton.
\end{definition}
% \begin{definition}[Total, Original, and Redundant Automata]
% Given the system Total Automaton (manufacturing layout) is defined by the user as $\mathcal{A}^{T}=(Q, C,\Sigma, E, I, R, P, q_{0})$, the Original Automaton $\mathcal{A}^{O} \subseteq \mathcal{A}^{T}$ is defined by $\mathcal{A}^{O}=(Q^{O}, C^{O},\Sigma, E^{O}, I^{O}, \allowbreak R^{O}, P^{O}, q_{0})$ (manufacturing layout without standby conveyor belts). Moreover, Redundant Automaton $\mathcal{A}^{R}=\cup_{v=1}^{k}\mathcal{A}^{R_{v}}$ is defined by $\mathcal{A}^{R}=\cup_{v=1}^{k}(Q^{R_{v}}, C^{R_{v}},\Sigma, E^{R_{v}}, I^{R_{v}}, R^{R_{v}},\allowbreak P^{R_{v}}, q_{R_{v}})$ (standby conveyor belts),  where $k$ is the number of redundant paths, $q_{R_{v}}$ and $Q^{R_{v}}$ are the redundant initial state and the set of redundant state locations defined as below:
% \begin{align}
%     & Q^{R_{v}}=\{ \langle q^{i},q^{r_{1}},\cdots,q^{r_{n}},q^{j}\rangle | q^{i},q^{j} \in Q^{O}, \\
%     &\qquad \quad \forall s \in \{r_{1},\cdots,r_{n}\}, \hspace{0.2cm} q^{s}\in Q^{T}-Q^{O} \} \notag\\
%     & q_{R_{v}}=q^{i}
% \end{align}
% Thus, $Q^{R}=\cup_{v=1}^{k}Q^{R_{v}}$ represents an ordered state set in which only the first and last states of all redundant paths exist in the original system automaton. The rest of the elements of $\mathcal{A}^{R}$ are user-defined. 
% \end{definition}
The following definitions determine the properties that two paths in the original automaton $A^{O}$ should share so that they can be connected by a redundant path from redundant automaton $\mathcal{A}^{R}$ in case of state failure in the path the redundant path is branching from.
%\begin{remark}
%It is assumed that the reader is familiar with basic concepts of the DES. For definitions of End Parity, Equivalent Paths, Legal State/Path, Emergency Declaration Signal and Update Operator which are frequently used in this work, see reference\cite{anbarani_risk-averse_2022}.
%\end{remark}
\begin{definition}[Active Redundant Path]
\label{def-activeredundantpath}
The redundant path $Q^{R}=\langle q^{i},q^{r_{1}},\cdots,q^{r_{n}},q^{j}\rangle$ that connects equivalent paths $U_{1}=\langle \cdots ,q^{i-1},q^{i},q^{i+1},\cdots \rangle$ and $U_{2}=\langle \cdots ,q^{j-1},q^{j},q^{j+1},\allowbreak \cdots \rangle$ is defined to be an active redundant path if the two subsets $U_{m1}=\langle \cdots,q^{i-1},q^{i}\rangle$ and $U_{m2}=\langle q^{j},q^{j+1},\cdots \rangle$ are legal paths. Otherwise, the redundant path is passive, i.e. the flexible material handling resources that is not used.
\end{definition}
\begin{definition}[Out-Degree Centrality \& Branch State]
The out-degree centrality of a given station $q^{i}$, denoted by $x_{i}$, is defined to be the number of the flexible material handling resources this station is connected to \cite{newman_networks_2018}. 
% This is equivalent to the number of legal paths with state $q^{i}$ being its initial state.
If the station $q^{i}$ has an out-degree centrality of more than one, it is called Branch State.   
\end{definition}
% \begin{definition}[Out-Degree Centrality \& Branch State]
% The $x_{i}$, the out-degree centrality of a given state $q^{i}$ is defined to be the number of edges starting from that state and ending in a legal state\cite{newman_networks_2018} (number of conveyor belts a machining station is connected to). This is equivalent to the number of legal paths with state $q^{i}$ being its initial state. If the state $q^{i}$ has an out-degree centrality of more than one, it is called Branch State.   
% \end{definition}
\begin{definition}[Path Length]
For a given path $U = \langle q^{i},q^{i+1},\cdots,q^{i+n}\rangle \in Q$ the path length $L(U)=n$ is defined to be the number of machining stations or buffers in a production line.
\end{definition}
% \begin{definition}[Path Length]
% For path $U = \langle q^{i},q^{i+1},\cdots,q^{i+n}\rangle \in Q$ the path length $L(U)=n$ is defined to be the number of legal states in the path (number of machining stations within a production line).
% \end{definition}
\begin{definition}[Committed Sub-Path]
 Given a path $U = \langle q^{i},q^{i+1},\cdots,q^{i+n}\rangle \in Q$, the path $U_{CSP} = \langle q^{i+l_{1}},\cdots,q^{i+l_{2}}\rangle,0\leq l_{1}<l_{2} \leq n$  is called the Committed Sub-Path (CSP) if:
 \begin{align}
 \label{eq-outdegreecentrality}
    & x_{i+l_{1}},x_{i+l_{2}}\geq 2, \notag\\
    & \forall l_{1}<j<l_{2}: x_{i+j} = 1 
 \end{align}
Therefore, CSP is a portion of the production line where all stations have an out-degree centrality of one, i.e., no branch state exists in that portion.
\end{definition}

% \begin{definition}[Committed Sub-Path]
%  Given an arbitrary path $U = \langle q^{i},q^{i+1},\cdots,q^{i+n}\rangle \in Q$, the path $U_{CSP} = \langle q^{i+l_{1}},\cdots,q^{i+l_{2}}\rangle,0\leq l_{1}<l_{2} \leq n$  is called the Committed Sub-Path (CSP) if:
%  \begin{align}
%  \label{eq-out degree centrality}
%     & x_{i+l_{1}},x_{i+l_{2}}\geq 2, \notag\\
%     & \forall l_{1}<j<l_{2}: x_{i+j}=0 
%  \end{align}
% Therefore, CSP is a portion of the path (production line) where all states (machining stations) have an out-degree centrality of one (only connected to one conveyor belt at their output).
% \end{definition}
\begin{definition}[Path Commitment Measure]
Assume $\mathcal{H}_{U}=\{ U_{1},\cdots,U_{m} \}$ is the set of all $U_{CSP}$ of path $U = \langle q^{i},q^{i+1},\cdots,q^{i+n}\rangle \in Q$. Then $\kappa(U)$, the Path Commitment Measure (PCM) for path $U$, is defined as:
\begin{align}
\label{eq-kappa}
    \kappa(U)=\frac{\Gamma(\mathcal{H}_{U})}{mL(U)}
\end{align}
where $\Gamma(\mathcal{H}_{U})=\sum_{i=1}^{m} L(U_{i})$ is the total length of all CSPs in $U$.
If $U$ does not have at least two Active Redundant Paths as per Definition \ref{def-activeredundantpath}, then $\kappa(U)$ is defined to be 1.
\par
The PCM represents the average length of all CSP existing in the path introducing a measure of the chance of failure when the current station is located in one of CSPs in path~$U$. As $0\leq \kappa(U) \leq 1$, the higher values of $\kappa(U)$ represent a higher risk of failure.
\end{definition}
% \begin{definition}[Path Commitment Measure]
% Assume $\mathcal{H}_{U}=\{ U_{1},\cdots,U_{m} \}$ is the set of all CSPs of the arbitrary path $U = \langle q^{i},q^{i+1},\cdots,q^{i+n}\rangle \in Q$. Then, the Path Commitment Measure (PCM) for path $U$ is defined as:
% \begin{align}
% \label{eq-kappa}
%     \kappa(U)=\frac{\Gamma(\mathcal{H}_{U})}{mL(U)}
% \end{align}
% where $\Gamma(\mathcal{H}_{U})=\sum_{i=1}^{m} L(U_{i})$ is the total length of all CSPs in $U$.\eb{Instead of a subscript, consider making it $\kappa(U)$ so that the dependence on the path is more explicit and similar to the rest of the notation.}
% If $U$ does not have at least two Active Redundant Paths as per Definition \ref{def-active redundant path}, then $\kappa(U)$ is defined to be 1.
% \par
% The PCM represents the average length of all CSP existing in the path introducing a measure of the chance of failure when the current state is located in one of CSPs in path $U$. As $0\leq \kappa(U) \leq 1$, the higher values of $\kappa(U)$ represent a higher risk of failure.
% \end{definition}
\begin{remark}
    $\kappa(U)=0$ corresponds to a completely risk-averse case where all states in the path $U$ are branch states.
\end{remark}
\begin{theorem}
If two paths $U^{1}$ and $U^{2}$ have the same length, the same number of CSPs, and $\Gamma(\mathcal{H}_{U^1})=\Gamma(\mathcal{H}_{U^2})$, then $\kappa({U^{1}})=\kappa({U^{2}})$.
\end{theorem}
%\begin{corollary}
%For a given path, the order by which one CSP is connected to other CSPs does not affect the PCM.
%\end{corollary}
%\begin{proof}
%Assume two paths $U^1$ and $U^2$ with the same length and identical CSPs connected in a different order, i.e. $m_{1}=m_{2}$ and $L(U^1)=L(U^2)$. As the order of CSPs does not affect the path length, $\Gamma(H_{U^1})=\Gamma(H_{U^2})$. Therefore, $\kappa({U^1})=\kappa({U^2})$ and the Lemma is established.
%\end{proof}
\begin{corollary}
If two paths $U^{1}$ and $U^{2}$ have the same length $L(U)$ and the same number of branch states then the out-degree centrality of branch states does not affect the PCM.
\end{corollary}
\begin{proof}
    $\Gamma(\mathcal{H}_{U^1})=\Gamma(\mathcal{H}_{U^2})=L(U)-w$ where $w$ is the number of branch states. $\Gamma(\mathcal{H}_{U^1})$ and $\Gamma(\mathcal{H}_{U^2})$ do not depend on the out-degree centrality of branch states. Using Equation~\ref{eq-kappa}, $\kappa({U^1})=\kappa({U^2})$ and the Lemma is established.
\end{proof}
\begin{remark}
In this work, symbol "$\to$" is the notion of \textit{imply}.
\end{remark}
\begin{corollary}
\label{corollary-minkappa}
If two paths $U^{1}$ and $U^{2}$ have the same length $L(U)$, and $U^1$ has a larger count of branch states, then $U^1$ has a lower PCM.  
\end{corollary}
\begin{proof}
Assume the count of branch states of $U^1$ and $U^2$ are $m_1$ and $m_2$, respectively.
    \begin{align*}
      &m_{1}>m_{2} \to  L(U)-m_{1}<L(U)-m_{2}\\
      &\Gamma(\mathcal{H}_{U^1})<\Gamma(\mathcal{H}_{U^2}) \to \kappa({U^1})<\kappa({U^2}) 
    \end{align*}
    
\end{proof}
\qedhere
\begin{theorem}
\label{theorem-centrality2}
    Given both a path of the length $L(U)$ and the sum of branch state out-degree centralities of $\gamma$, the path with the maximum number of branch states of out-degree centrality of 2 has the lowest PCM. 
\end{theorem}
\begin{proof}
   Assume $B$ is the set of all branch states in a path. Then, the sum of branch states out-degree centralities is $\sum_{i\in B}x_{i}=\gamma$, where $x_{i}$ is the out-degree centrality of state $i$. According to Lemma \ref{corollary-minkappa}, a path with minimum PCM value shall have the largest possible number of branch states with the sum of $\gamma$. Therefore, the path with minimum PCM value has the largest number of elements of $B$ and the sum of out-degree centralities of the branches is constant, to maximize the number of branches their out-degree centralities has to be minimized whose value according to Equation~\ref{eq-outdegreecentrality} is 2. Therefore, the Statement is established.
\end{proof}
\begin{theorem}
\label{corollay-optimalriskmeasure}
    If path $U$ of length $L(U)$ and the sum of branch state out-degree centralities of $\gamma$ has the lowest PCM among all other paths of the same length and sum of branch state out-degree centralities, then adding/removing a CSP $U^{\prime} \subset U$ of length $L(U^{\prime})$ results in a new path which has the lowest PCM among all paths of the new length and sum of branch state out-degree centralities of $\gamma$.
\end{theorem}
\begin{proof}
Adding/Removing a CSP does not affect the sum of branch states out-degree centrality nor the number of branch states, therefore:
\begin{align*}
    &U^{\prime} \subset U \to U-U^{\prime} \subset U\\
    &\sum\limits_{\substack{i\in U-U^{\prime}}}^{}x_{i}=\gamma.
\end{align*}
As the out-degree centrality of all branch states of path $U-U^{\prime}$ are also 2, according to Statement \ref{theorem-centrality2} the path has the lowest PCM value and the Statement is established.
\end{proof}
\section{Risk-Averse PTA MPC}
\label{sec-PTAMPC}
%In PTA, the time elapses independent of any event.
The PTA MPC is a control algorithm in which the control objective is to find a path that includes all states in the desired states set in the order defined by the user while minimizing the cumulative temporal cost of state executions and edge transverses \cite{anbarani_risk-averse_2022}. The objective of the Risk-Averse PTA MPC algorithm is to minimize the risk of path failure due to state failure in addition to the objectives of the PTA-MPC algorithm.
%\textcolor{red}{The control objective for the PTA-MPC algorithm is to find a path that minimizes the sum of all temporal cost. For PTA-MPC, the time elapsed at each state execution and state transition is stored as a clock variable. The temporal cost is calculated based on these clock variables. As a result, a finite path, i.e. a finite ordered sequence of states, is associated with a deterministic temporal cost which enables formulating a constrained optimization problem.}
% The control objective is to find a path such that the sum of all cost metrics for states comprising the path is minimized.
\par
Figure \ref{figure-MPCblockdiagram} depicts the risk-averse PTA MPC block diagram. Prior to implementing the risk-averse PTA MPC block, the update operator updates the original PTA using the redundant PTA and sensory feedback \cite{anbarani_risk-averse_2022}. Also, the current state and the remaining desired state set are recalled from the memory to input the update operator and eventually run the risk-averse PTA MPC block. The block executes a multi-objective optimization in terms of cost and risk and returns the optimal path, from which the first state is executed and stored in memory for the next iteration. If the executed state belongs to the desired states set, the set is updated by removing the executed state. Otherwise, the desired state set remains unchanged for the next iteration. This process is continued until either the desired states set is exhausted, or there exists no solution to the control problem making the risk-averse PTA-MPC \textit{unsatisfiable} ({\fontfamily{qcr}\selectfont
UNSAT}).
%\textcolor{red}{Figure \ref{figure-MPC_block diagram} depicts the risk-averse PTA MPC block diagram. At each iteration the update operator generates the updated PTA using the original PTA, redundant PTA, and EDS. The current state and remaining desired state set are retrieved from memory and used by the update operator and the risk-averse PTA MPC algorithm, respectively. The risk-averse PTA MPC also receives  user defined risk-factors for each state. These risk-factors are designated as $h_{i}$   $\forall i\in \{1,\cdots,|Q^{T}|\}$ and determine the weight of failure risk at each state. Thus, the risk factors provide a quantitative measure to compare paths in terms of risk.}
\par
The goal of the multi-objective optimization problem is to minimize the cost of the path while minimizing the risk, which can be represented as follows:
\begin{equation}
\label{eq-costriskoptimization-general}
    \alpha^{\ast}= \underset{\alpha \in \mathcal{L}(\mathcal{A})}{\arg\min } \left(V(\alpha) =\Bigl[P(\alpha), \Bar{R}(\alpha) \Bigr]\right)
\end{equation}
where $\alpha$ is the path i.e. the sequence of states, $V(\alpha)$ is the objective function value, $P(\alpha)=\sum_{i=1}^{|\alpha|}P_{i}$ is the cost associated with the path, and $\Bar{R}(\alpha)$ is the \textit{Average Risk Measure} of the path, and $\mathcal{L}(\mathcal{A})$ is the set of all feasible paths in PTA $\mathcal{A}$. The average risk measure is comprised of two parts: (i) user-defined multiplier and, (ii) PTA architecture-inherited value.
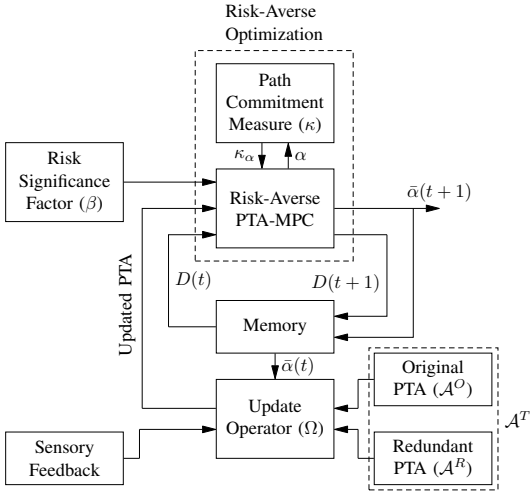
\begin{figure}[t]
\begin{center}
    \begin{adjustbox}{minipage=\textwidth,scale=0.7}
    \begin{tikzpicture}[>={Stealth[inset=0pt,length=8pt,angle'=28,round]},scale=1,font=\normalsize]
      % coordinates
        \coordinate (LLT) at (-1.65,-3);
        \coordinate (LLMPC) at (1,4);
        \coordinate (LLRPTA) at (4,-0.5);
        \coordinate (LLTPTA) at (3.9,-0.6);
        \coordinate (LLZ) at (-3,-0.5);
        \coordinate (LLLE) at (6,-1.75);
        \coordinate (LLOPTA) at (4,1);
        \coordinate (LLLN) at (6,0.75);
        \coordinate (LLUO) at (1,0);
        \coordinate (LLM) at (1,2);
        \coordinate (LLFP) at (-3,4.5);
        \coordinate (LLPCM) at (1,6);
        \coordinate (LLDB) at (0.6,3.75);
        \coordinate (LLDB2) at (-3.25,-0.75);
      % nodes
        \node[draw, minimum width=2.45cm, minimum height=2.7cm, anchor=south west, text width=2cm, font=\normalsize, align=center,label=0:$ \mathcal{A}^{T}$, densely dashed] (TPTA) at (LLTPTA)
        {};

        \node[draw, minimum width=3cm, minimum height=4cm, anchor=south west, text width=0.5cm, font=\normalsize, align=center,label={[align=left]90:{Risk-Averse\\Optimization}}, densely dashed] (DB) at (LLDB)
        {};

        %\node[draw, minimum width=6.85cm, minimum height=2.65cm, anchor=south west, text width=0.5cm, font=\normalsize, align=center,label={[align=left]180:\rotatebox{90}{Update Mechanism \cite{anbarani_risk-averse_2022}}}, densely dashed] (DB2) at (LLDB2)
        %{};
        
        \node[draw, minimum width=2cm, minimum height=1.5cm, anchor=south west, text width=2cm, font=\normalsize, align=center] (FP) at (LLFP)
        {Risk Significance\\Factor ($\beta$)};

        \node[draw, minimum width=2cm, minimum height=1.5cm, anchor=south west, text width=2cm, font=\normalsize, align=center] (PCMB) at (LLPCM)
        {Path Commitment\\Measure ($\kappa$)};

        \node[draw, minimum width=2cm, minimum height=1cm, anchor=south west, text width=2cm, font=\normalsize, align=center] (EDS) at (LLZ)
        {Sensory Feedback};
      
        \node[draw, minimum width=2cm, minimum height=1.5cm, anchor=south west, text width=2cm, font=\normalsize, align=center] (MPC) at (LLMPC)
        {Risk-Averse\\PTA-MPC};
        
        \node[draw, minimum width=2cm, minimum height=1cm, anchor=south west, text width=2cm, font=\normalsize, align=center] (LT) at (LLOPTA) {Original\\PTA  ($\mathcal{A}^{O}$)};
        
        \node[draw, minimum width=2cm, minimum height=1cm, anchor=south west, text width=2cm, font=\normalsize, align=center] (Z) at (LLRPTA) {Redundant\\PTA ($\mathcal{A}^{R}$)};
        
        \node[draw, minimum width=2cm, minimum height=1.5cm, anchor=south west, text width=2cm, font=\normalsize, align=center] (UO) at (LLUO) {Update\\Operator ($\Omega$)};
        
        \node[draw, minimum width=2cm, minimum height=1cm, anchor=south west, text width=2cm, font=\normalsize, align=center] (M) at (LLM) {Memory};
        
        %edges
        % MPC
        \draw[->] (MPC.0) -- ($(MPC.0)+(2,0)$) node[anchor=center,above] {$\bar{\alpha}(t+1)$};
        \draw[->] ($(MPC.90)+(0.25,0)$) --  ($(PCMB.270)+(0.25,0)$)node[anchor=center,right,pos=0.5] {$\alpha$};
        \draw[->] ($(PCMB.270)+(-0.25,0)$) --  ($(MPC.90)+(-0.25,0)$)node[anchor=center,left,pos=0.5] {$\kappa_{\alpha}$};
        \path[draw,-] ($(MPC.0)+(0,-0.5)$) -- ($(MPC.0)+(1,-0.5)$);
        \draw[->] ($(MPC.0)+(1,-0.5)$) |- node[left,pos=0.3]{$D(t+1)$} ($(M.0)+(0,0.2)$);
        \draw[->] ($(MPC.0)+(1.5,0)$) |-  ($(M.0)+(0,-0.2)$);
        \draw[->] ($(M.270)$) -- (UO.90)node[right,pos=0.5] {$\bar{\alpha}(t)$};
        
        \path[draw,-] ($(UO.180)+(0,0.2)$) -- ($(UO.180)+(-1.4,0.2)$);
        \draw[->] ($(UO.180)+(-1.4,0.2)$) |- node[left,pos=0.26]{\rotatebox{90}{Updated PTA}} (MPC.180);
        
        % Memory
        \path[draw,-] (M.180) -- ($(M.180)+(-0.9,0)$);
        \draw[->] ($(M.180)+(-0.9,0)$) |- node[right,pos=0.25]{$D(t)$} ($(MPC.180)+(0,-0.5)$);
        
        % Risk Factor
        
        \draw[->] (FP.0) -- ($(MPC.180)+(0,0.5)$);
        
        %Languages
        \path[draw,-] (LT.180) -- ($(LT.180)+(-0.3,0)$);
        \path[draw,-] ($(LT.180)+(-0.3,0)$) -- ($(LT.180)+(-0.3,-0.55)$);
        \draw[->] ($(LT.180)+(-0.3,-0.55)$) -- ($(UO.0)+(0,0.2)$);
        
        \path[draw,-] (Z.180) -- ($(Z.180)+(-0.3,0)$);
        \path[draw,-] ($(Z.180)+(-0.3,0)$) -- ($(Z.180)+(-0.3,0.55)$);
        \draw[->] ($(Z.180)+(-0.3,0.55)$) -- ($(UO.0)+(0,-0.2)$);
        
        \path[draw,-] (EDS.0) -- ($(EDS.0)+(0.3,0)$);
        \path[draw,-] ($(EDS.0)+(0.3,0)$) -- ($(EDS.0)+(0.3,0.55)$);
        \draw[->] ($(EDS.0)+(0.3,0.55)$) -- ($(UO.180)+(0,-0.2)$);
        
    \end{tikzpicture}    
    \end{adjustbox}
    \end{center}
    \caption{Block diagram of risk-averse MPC with risk significance factor ($\beta$).}
    \label{figure-MPCblockdiagram}
\end{figure}
The user-defined multiplier indicates the relative significance of the average risk measure compared to the price whereas the PTA architecture-inherited value determines the risk of failure due to the existing CSPs in a particular path.
%\textcolor{red}{\begin{equation}
%\label{eq-cost & risk optimization-general}
%\alpha^{\ast}= \underset{\alpha \in \mathcal{L}(\mathcal{A})}{\arg\min } \left(V(\alpha) =\sum_{i=1}^{|\alpha|}P_{i}+\sum_{i=1}^{|\alpha|}R_{i}\right),  
%\alpha^{\ast}= \underset{\alpha \in \mathcal{L}(\mathcal{A})}{\arg\min } \left(V(\alpha) =\sum_{i=1}^{|\alpha|}P_{i},\sum_{i=1}^{|\alpha|}R_{i}\right),  
%\end{equation}
%where $\alpha$ is the path, i.e. the sequence of states, and $R_{i}$ is the \textit{Risk Measure} of the $i^{th}$ state showing the associated risk weight corresponding to each state, and $\mathcal{L}(\mathcal{A})$ is the set of all feasible paths in PTA $\mathcal{A}$. The risk measure is comprised of two parts; user-defined weights and PTA-inherited weight. The user-defined weights indicate the risk chance as determined by the user whereas the PTA-inherited weight is based on the number of redundant paths available.}
\par
Equation~\ref{eq-costriskoptimization-general} constitutes a multi-objective optimization problem in terms of cost and risk. By describing the average risk measure in terms of PCM $\kappa({\alpha})$ and cost, i.e. $\Bar{R}(\alpha)=\kappa({\alpha})\sum_{i=1}^{|\alpha|}P_{i}$, one may convert the optimization problem into a single-objective problem in terms of cost only. The resulting optimization problem is as follows:
\begin{align}
& \alpha^{\ast}= \underset{\alpha \in \mathcal{L}(\mathcal{A})} {\arg\min } \left(V(\alpha) =\Bigl[1+\beta \kappa({\alpha})\Bigr] \sum_{i=1}^{|\alpha|}P_{i}\right).
\end{align}
%\textcolor{red}{While Equation~\ref{eq-cost & risk optimization-general} can be solved as a multi-objective optimization problem, it can be reduced to a single-objective optimization problem if the risk measures can be described with respect to the state cost, i.e. $R_{i}=u_{i}P_{i}$. In this manner, the risk measure can be expressed relative to the cost measure as follows:
%\begin{align}
%& \alpha^{\ast}= \underset{\alpha \in \mathcal{L}(\mathcal{A})} {\arg\min } \left( V(\alpha) = \sum_{i=1}^{|\alpha|}P_{i}+\sum_{i=1}^{|\alpha|}u_{i}P_{i}\right). \label{eq-cost function}
% & subject\hspace{0.2 cm}to: \alpha \in \mathcal{L}(\mathcal{A}) \notag
%\end{align}
%}
%\textcolor{red}{The variable $u_{i}$ captures the \textit{uncertainty ratio} associated with $i^{th}$ state of the path which is defined as follows:
%\begin{equation}
%\label{eq-uncertainty ratio}
%    u_{i}=\frac{h_{i}}{x_{i}}
%\end{equation}
%where $h_{i}$ is the risk factor assigned by the user to each state and $x_{i}$ is \textit{Out-degree Centrality Measure} defined as being the number of outgoing edges of a given state .}
where $0\leq \beta \leq 1$ is the user-defined risk significance factor which determines the relative significance of risk with respect to the cost. When $\beta=0$, the optimization problem is simplified as a PTA-MPC without the risk-averse feature while $\beta=1$ translates to the equivalent significance of cost and risk. All cases with $\beta >1$ represent a higher emphasis on risk-averseness than the cost. For non-deterministic systems where there are infinitely many states or when $\beta \to \infty$, the multi-objective optimization reduces to a single-objective optimization in terms of risk-averseness only. However, in the case of $\beta \to \infty$, the optimization problem is undefined and, therefore, it is not used in the proposed framework to model and control manufacturing systems.
\par
% Since the last state of each path is the last desired state in the desired state set $\Sigma$ and thus, exists in all possible solution and has a zero out-degree centrality, it is excluded from the risk sum.
%We exclude the last state of each path from the risk sum since it is assumed to have centrality zero and it belongs to every possible path.
To impose the fact that the solution belongs to the legal automaton, i.e. it only includes the states that are legal, in addition to the cost objective, we introduce the following constraints: (i) the optimal path belongs to the set of all legal paths in the PTA, and; (ii) the path satisfies the order of the desired state set $\Sigma$.
% h the desired states occur in the path has to be exactly the same as that in the desired states set $\Sigma$.
Assuming that $\alpha$ is a path of length N, i.e. $\alpha=\langle \alpha_{i_{1}},\cdots, \alpha_{i_{N}} \rangle$ and $\Sigma=\langle \sigma_{1},\cdots,\sigma_{N_{d}} \rangle$ with $N_{d}\leq N$, where $N_{d}$ is the number of desired states, the constrained optimization problem can be described as:
\begin{align}
\label{eq-costfunction-modified}
& \alpha^{\ast}= \underset{} {\arg\min } \left( V(\alpha) = \Bigl[1+\beta \kappa({\alpha})\Bigr]\sum_{i=1}^{|\alpha|}P_{i}\right)\\
&\alpha \in \mathcal{L}(\mathcal{A}) \label{eq-validity}\\
&\phi(\Sigma)=\text{True}
\end{align}
where $\phi(\Sigma)$ is True if:
\begin{align}
&\forall m<n \in \{1,\cdots,N_{d}\},\forall \sigma_{m},\sigma_{n} \in \Sigma \notag\\
&\exists \alpha_{ik},\alpha_{il} \in \alpha: \sigma_{m}=\alpha_{ik},\sigma_{n}=\alpha_{il},k<l\hspace{0.5cm}\label{eq-desiredvalue}
\end{align}
and is False otherwise. 
%\textcolor{red}{
%\begin{align}
%\label{eq-cost function- modified}
%&\alpha^{\ast}= \underset{\alpha}{\arg\min} \left( V(\alpha)= \sum_{i=1}^{|\alpha|-1}(1+\frac{h_{i}}{x_{i}})P_{i}\right)\\
%&\alpha \in \mathcal{L}(\mathcal{A}) \label{eq-validity}\\
%&\forall m<n \in \{1,\cdots,N_{d}\},\forall \sigma_{m},\sigma_{n} \in \Sigma \notag\\
%&\exists \alpha_{ik},\alpha_{il} \in \alpha: \sigma_{m}=\alpha_{ik},\sigma_{m}=\alpha_{il},k<l\hspace{0.5cm}\label{eq-desired value}
%\end{align}
%}

The form of constrained optimization problem described by Equation~\ref{eq-costfunction-modified} is similar to that detailed in \cite{balta_model_2022}. Thus, it can be converted to a first-order logic problem for which theorem provers can be used to find a solution \cite{bjorner_z_nodate}.
%\textcolor{red}{Equation~\ref{eq-cost function- modified} describes the risk measure in term of cost. Equation~\ref{eq-validity} ensures that the solution is among feasible paths of $\mathcal{A}$. Equation~\ref{eq-desired value} ensures that the solution includes all desired states according to $\Sigma$.Regardless of the value of EDS, at all times the PTA and its constraints are the same as in \cite{balta_model_2022}. Thus, the optimization problem described by Equation~\ref{eq-cost function- modified} is reduced to first-order logic, which can be solved using the $Z_{3}$ theorem prover \cite{bjorner_z_nodate}.}
\begin{algorithm}[t]
\caption{Risk-Averse PTA-MPC Optimization Algorithm}
\label{algorithm-MPC}
\begin{algorithmic}[1]
\Require{$\mathcal{A}=(Q, C,\Sigma, E, I, R, P, q_{0}),D(t),\beta, \kappa({\alpha})$}
\Ensure{$\alpha,\alpha(t+1),D(t+1)$ or {\fontfamily{qcr}\selectfont UNSAT}}
\Initialize{
$\mathcal{A}_{updated} \gets \mathcal{A}$\\
$D(t+1) \gets \Sigma$,\hspace{0.1cm}$\alpha(t+1) \gets \emptyset$\\
$\bar{\alpha}(t) \gets q_{0}$,\hspace{0.1cm}$x \gets \emptyset$\\
}
\While{$D(t+1)\neq \emptyset$}
\State (i) Check if Eq. \ref{eq-costfunction-modified} is satisfiable
\If{(i) is not {\fontfamily{qcr}\selectfont UNSAT}}
\State $V_{min} \gets \emptyset$
\While{$T\neq \emptyset$ (The set of solutions)}
\State Pick one path $\alpha \in T$
\State Calculate the Path Commitment Measure $\kappa({\alpha})$ using Equation \ref{eq-kappa}
\State Determine risk-averse PTA-MPC optimization solution $\bar{\alpha}(t+1)$ based on Eq. \ref{eq-costfunction-modified} using $\mathcal{A}_{updated}$ and $\kappa({\alpha})$ and $\beta$.
\If{$V(\alpha)<V_{min}$}
\State $V_{min} \gets V(\bar{\alpha})$
\EndIf
\EndWhile

\State $\alpha(t+1) \gets \langle \alpha(t),\bar{\alpha}[1](t+1) \rangle$
\State \Comment{$\bar{\alpha}[1](t+1)$ is the first element of $\bar{\alpha}(t+1)$}
\Else
\State return {\fontfamily{qcr}\selectfont UNSAT}
\State stop
\EndIf
\State $\bar{\alpha}(t) \gets \bar{\alpha}[1](t+1)$
\If{$\bar{\alpha}[1](t+1) \in \Sigma$}
\State $D(t+1) \gets \Sigma-{\bar{\alpha}[1](t+1)} $
\EndIf
\State $D(t) \gets D(t+1)$
\State return $\alpha(t+1),D(t+1)$
\EndWhile
\end{algorithmic}
\end{algorithm}
Algorithm \ref{algorithm-MPC} depicts the steps of risk-averse optimization process.
\begin{figure*}[t]
    \centering
    \subfloat[t][]{\includegraphics[scale=0.4,valign=b]{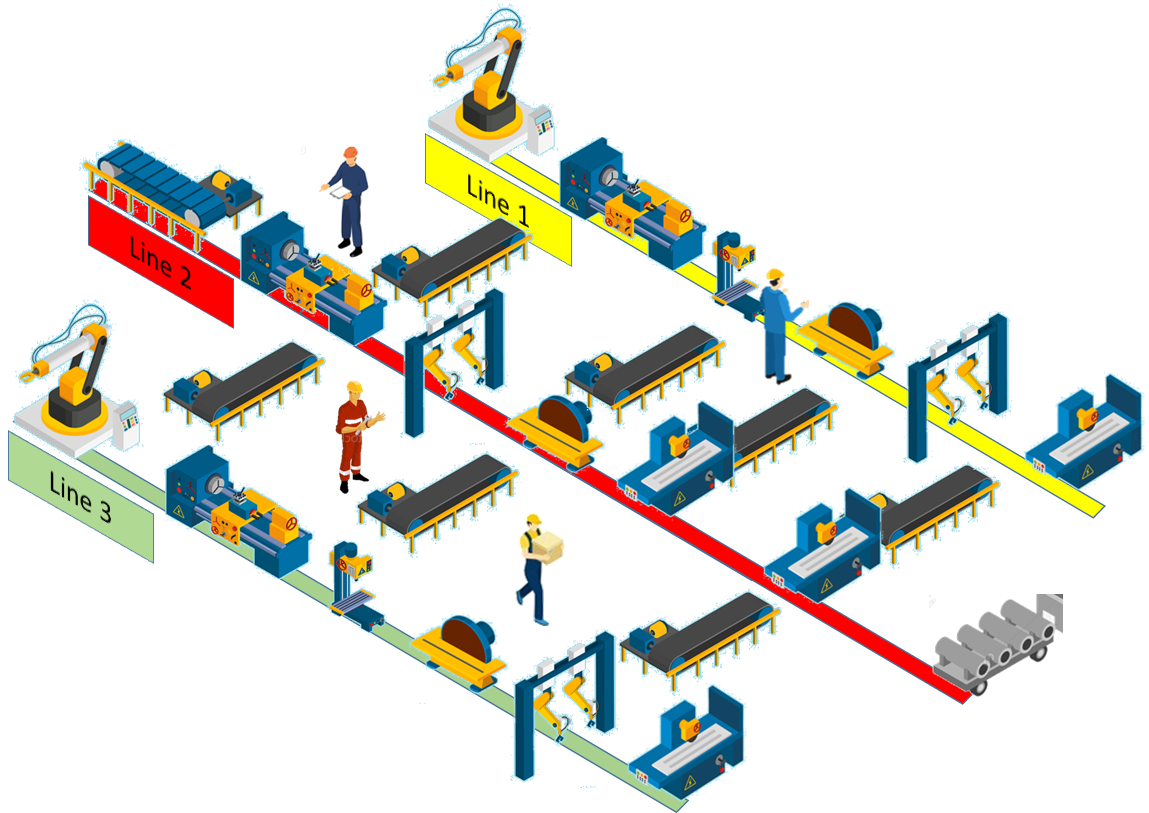}\label{figure-plantschematic}}
    \subfloat[t][]{
    \begin{adjustbox}{scale=0.7,valign=b}
    \begin{tikzpicture}[node distance = 2cm, on grid,>={Stealth[inset=0pt,length=6pt,angle'=28,round]}]
    
    \node[state, initial] (q1) {$q^1$};
    \node[state] at (1,-2.5) (q2) {$q^2$};
    \node[state, right of=q2] (q3) {$q^3$};
    \node[state, right of=q3] (q4) {$q^4$};
    \node[state, right of=q4] (q5) {$q^5$};
    \node[state, right of=q5] (q6) {$q^6$};
    \node[state, right of=q6] (q7) {$q^7$};
    \node[state] at (11.8,0) (q8) {$q^8$};

    \node[state] at (2,0) (q9) {$q^9$};
    \node[state, right of=q9] (q10) {$q^{10}$};
    \node[state, right of=q10] (q11) {$q^{11}$};
    \node[state, right of=q11] (q12) {$q^{12}$};
    \node[state, right of=q12] (q13) {$q^{13}$};
    
    \node[state] at (1,2.5) (q14) {$q^{14}$};
    \node[state, right of=q14] (q15) {$q^{15}$};
    \node[state, right of=q15] (q16) {$q^{16}$};
    \node[state, right of=q16] (q17) {$q^{17}$};
    \node[state, right of=q17] (q18) {$q^{18}$};
    \node[state, right of=q18] (q19) {$q^{19}$};
    
    \node[state,fill=red!20, densely dashed] at (3,-1.25) (q20) {$q^{20}$};
    \node[state,fill=red!20, densely dashed] at (7,-1.25) (q21) {$q^{21}$};
    \node[state,fill=red!20, densely dashed] at (11,-1.25) (q22) {$q^{22}$};
    \node[state,fill=red!20, densely dashed] at (2,1.25) (q23) {$q^{23}$};
    \node[state,fill=red!20, densely dashed] at (5,1.25) (q24) {$q^{24}$};
    \node[state,fill=red!20, densely dashed] at (7,1.25) (q25) {$q^{25}$};
    \node[state,fill=red!20, densely dashed] at (11,1.25) (q26) {$q^{26}$};

 \draw 
    
   (q1) edge[bend right] node [below left]{} (q2)
   (q2) edge node [below] {} (q3)
   (q3) edge node [below]{} (q4)
   (q4) edge node [below]{} (q5)
   (q5) edge node [below]{} (q6)
   (q6) edge node [below]{} (q7)
   (q7) edge[bend right] node [below right] {} (q8)
   
   (q1) edge node [above]{} (q9)
   (q9) edge node [above]{} (q10)
   (q10) edge node [above]{} (q11)
   (q11) edge node [above]{} (q12)
   (q12) edge node [above]{} (q13)
   (q13) edge node [above]{} (q8)
   
   (q1) edge[bend left] node [below left]{} (q14)
   (q14) edge node [above]{} (q15)
   (q15) edge node [above]{} (q16)
   (q16) edge node [above]{} (q17)
   (q17) edge node [above]{} (q18)
   (q18) edge node [above]{} (q19)
   (q19) edge[bend left] node [above right]{} (q8)
   
   (q2) edge[bend right, densely dashed] node [above left]{} (q20)
   (q20) edge[bend right, densely dashed] node [above left]{} (q9)
   (q4) edge[bend right, densely dashed] node [right]{} (q21)
   (q21) edge[bend right, densely dashed] node [right]{} (q11)
   (q6) edge[bend right, densely dashed] node [right]{} (q22)
   (q22) edge[bend right, densely dashed] node [right]{} (q13)
   (q15) edge[bend right, densely dashed] node [right]{} (q23)
   (q23) edge[bend right, densely dashed] node [right]{} (q10)
   (q15) edge[bend left, densely dashed] node [right]{} (q24)
   (q24) edge[bend right, densely dashed] node [right]{} (q11)
   (q18) edge[bend right=10, densely dashed] node [above right]{} (q25)
   (q25) edge[bend right, densely dashed] node [above right]{} (q12)
   (q18) edge[bend left, densely dashed] node [above right]{} (q26)
   (q26) edge[bend left, densely dashed] node [above right]{} (q13);
    \end{tikzpicture}
    \end{adjustbox}
    \label{figure-plantPTA}}
    
    \caption{\protect\subref{figure-plantschematic} The schematic and,  \protect\subref{figure-plantPTA} equivalent PTA with Original Automaton ($\mathcal{A}^{O}$ solid line) and Redundant Automaton ($\mathcal{A}^{R}$ dashed line) for a production line.}
    \label{figure-samplePTA}
\end{figure*}
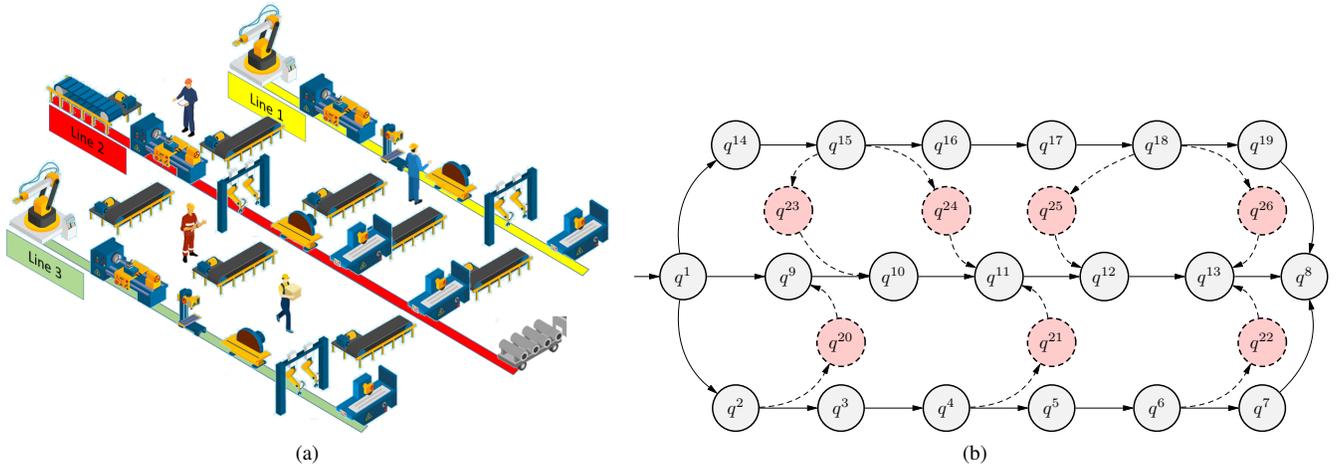
\par
In line 3, the PTA is updated by Update Operator after each event, then the current state is recalled from the memory to form the updated PTA \cite{anbarani_risk-averse_2022}. The value of $\beta$ and the current state name is used as input to the risk-averse optimization algorithm. In lines 4-7, all possible legal paths that contain the remaining desired states set are found and the corresponding branch states, CSPs, and CSP lengths are used to calculate each path PCM in line 8,  separately. The problem is treated as a single-objective PTA MPC and the objective function value $V({\alpha})$ is calculated for each path. Finally, in lines 10-12, a path with minimum objective function value is chosen, and in line 20, the first element of its respective path is executed. If the executed state belongs to the remaining desired states set, the set is updated in lines 21-23 by removing that state, and the results are reported in lines 24-25. The algorithm continues until either the desired states set becomes empty, or no feasible solution exists ({\fontfamily{qcr}\selectfont UNSAT}).

\section{Case Study}
\label{sec-casestudy}
Similar to \cite{kovalenko_priced_2020,kovalenko_cooperative_2022,ocker_framework_2019}, the production line exchange can be represented as redundant routes which are enabled in case a state failure occurs.
Thus, the proposed risk-averse PTA-MPC framework, which we call the \textit{proposed framework} hereinafter, updates the PTA in case of any \textit{event}, e.g., the completion of a job in a workstation or transitions of a part to other workstations.
The PTA is also updated in order to reflect changes due to failure prior to control actions. In case of workstation failure, the proposed framework enables the redundant routes which are closest to the current workstation, enabling the product to switch to other production lines to avoid failure. Note that although repetitive layout updates inherent in the MPC control scheme provide insight into \textit{current} workstation failures, these updates does not guarantee a fail-safe passage despite \textit{possible} workstation failures in the future. \begin{table}[h]
    \centering
    \renewcommand{\arraystretch}{1.2}
    \caption{List of the states, physical locations and centralities of manufacturing plant of Figure \ref{figure-samplePTA}.}
    \label{tab:listoflocationsforsystems}
    \begin{tabular}{|c|l|c|}
         \hline
         \multirow{2}{*}{State}&
         \multirow{2}{*}{Location}&
         \multirow{2}{*}{Centrality($x_{i}$)} \\
         &&\\
         \hline
         $q^1$& Material Depot & 3\\
         \hline
         $q^{2},q^{14}$& Robotic Manipulator 1 and 3 & 2,1\\
         \hline
         $q^{3},q^{9},q^{15}$& Turning Center 1,2 and 3 & 1,1,3\\
         \hline
         $q^{4},q^{16}$& Manual Handler 1 and 3 & 2,1\\
         \hline
         $q^{5},q^{11},q^{17}$& Disk Sanding Machine 1,2 and 3 & 1,1,1\\
         \hline
         $q^{6},q^{10},q^{18}$& CMM 1,2 and 3 & 2,1,3\\
         \hline
         $q^{7},q^{12},q^{19}$& Grinding Center 1,2,3 & 1,1,1\\
         \hline
         $q^{13}$& Grinding Center (standby) & 1\\
         \hline
         $q^8$& Target Storage & 0\\
         \hline
         $q^{20}-q^{26}$& Conveyor Belt & All 1\\
         \hline
    \end{tabular}
    \vspace{-2mm}
\end{table} %\input{figure-case-study-paths.tex}
\begin{figure*}[thpb]
\centering
    \includegraphics[width=0.55\textwidth]{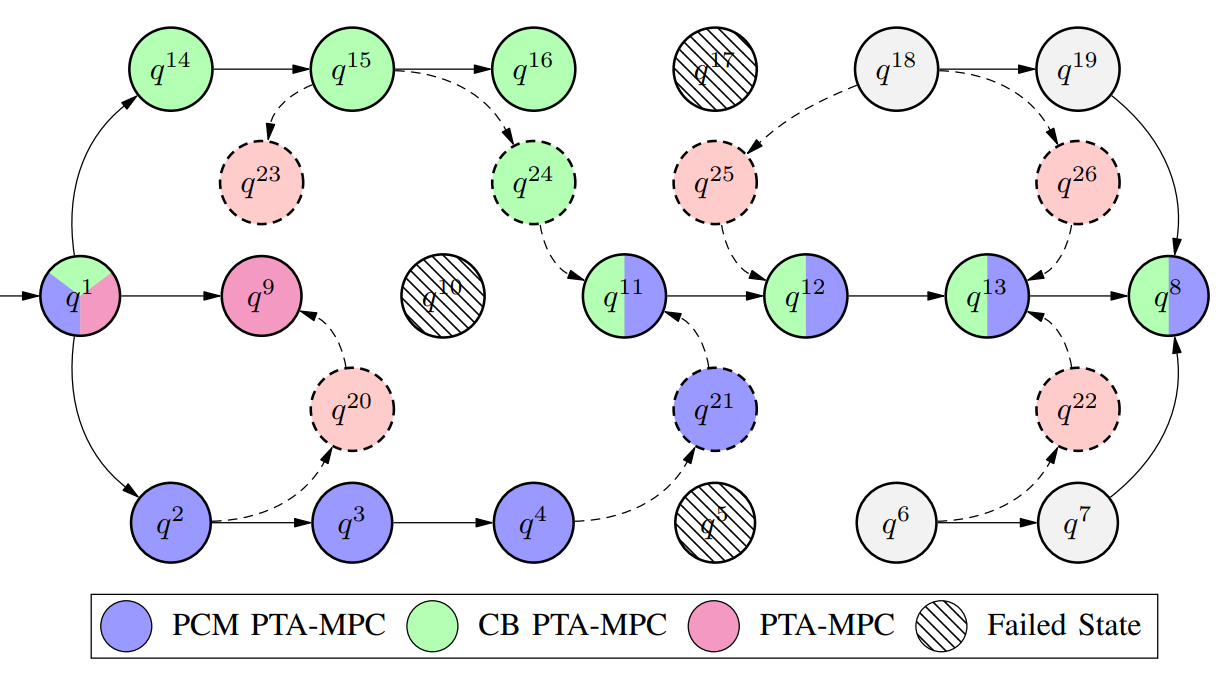}
	\caption{The solution for PCM risk-averse PTA-MPC (purple), Centrality-based risk-averse PTA-MPC (green), and regular PTA-MPC (red) under unpredictable state failure (hashed black) for a sample production line.}
 \label{fig:xxxx}
\end{figure*}
Risk-averseness is maintained by finding a path that contains the highest average number of redundant paths per path length to ensure the safe passage of the product in case any of the state ahead fails.

Figure~\ref{figure-plantschematic} depicts a typical layout of the flexible manufacturing plant used for customized car painting. Figure~\ref{figure-plantPTA} depicts the equivalent PTA in which original automaton $\mathcal{A}^{O}$ and redundant automaton $\mathcal{A}^{R}$ are marked by solid lines and dashed lines, respectively. The plant has three original routes $\alpha_{L1}=\langle q^{1},q^{14},\cdots,q^{19},q^{8} \rangle$, $\alpha_{L2}=\langle q^{1},q^{9},\cdots,q^{13},q^{8} \rangle$, $\alpha_{L3}=\langle q^{1},q^{2},\cdots,q^{7},q^{8} \rangle$ which are called \textit{Line 1}, \textit{Line 2} and \textit{Line 3}, respectively. Table \ref{tab:listoflocationsforsystems} summarizes the locations and corresponding out-degree centralities. For this example, $\Sigma=\{\sigma_{1}^{d}\}=\{q^{8}\}$ and $\beta=1$, i.e. the risk is equally important as cost and $\forall q^{i}\in Q: P_{i}=1,h_{i}=1$ where $h_{i}$ is the state risk-factor as in \cite{anbarani_risk-averse_2022}. Hereinafter, we designate the time at which state $q^{i}$ is occupied by $t_{i}$. We compare the result of our proposed method to that in \cite{anbarani_risk-averse_2022} and \cite{balta_model_2022} for two scenarios:
\begin{itemize}[leftmargin=*]
    \item Scenario 1:
    Initial $q^{10}$ failure at $t>t_{1}$ followed by $q^{5}$ failure at $t_{2}<t<t_{4}$ and $q^{17}$ failure at $t>t_{15}$,
    \item Scenario 2: 
    Initial $q^{10}$ failure at $t>t_{1}$ followed by $q^{5}$ failure at $t_{2}<t<t_{4}$ and $q^{17}$ failure at $t\leq t_{15}$.
\end{itemize}
In both Scenarios, there is a failure in workstation $q^{10}$ after a product has committed to Line 2 followed by a failure in station $q^5$ and $q^{17}$ when a product is in Line 1 or Line 3, respectively. What distinguishes Scenario 1 and Scenario 2 is the time at which failure in $q^{17}$ occurs. In Scenario 1, the failure occurs \textit{after} the product has left workstation $q^{15}$ whereas for Scenario 2, the failure occurs \textit{before} product is leaving $q^{15}$.
Figure~\ref{figure-samplePTAunderfailure} compares our  \textit{PCM risk-averse PTA-MPC} (PCM PTA-MPC) performance with PTA-MPC without risk-averse feature ~\cite{balta_model_2022}, denoted as \textit{PTA-MPC} (PTA-MPC), as well as the MPC formulation described in \cite{anbarani_risk-averse_2022}, denoted as \textit{Centrality-based risk-averse PTA-MPC} (CB PTA-MPC).
\par
In Scenario 1, since Line 2 has the least number of states, it is the cost-optimal solution. Therefore, the PTA-MPC selects Line 2 as its solution.
As the first failure occurs at $q^{10}$, the optimization fails and it returns {\fontfamily{qcr}\selectfont UNSAT}  which highlights the inflexibility of PTA-MPC to failures. Given that Line 1 has the largest number of out-degree centralities among all production lines, CB PTA-MPC takes Line 1. But as failure occurs at $q^{17}$, the optimization fails and it returns {\fontfamily{qcr}\selectfont UNSAT}. Therefore, despite having the largest centrality, Line 1 is still prone to failure. This is because Line 1 has few branch states with large out-degree centrality. Note that the most risk-averse design for Line 1 with the same out-degree centrality would have been to maximize the possible number of branch states with out-degree centrality of 2 (see Statement~\ref{theorem-centrality2}). For PCM PTA-MPC, Line 3 is taken as it has the lowest PCM value, i.e. $\kappa({L3})=0.15625$. As failure occurs at $q^{5}$, the redundant path is enabled allowing for rerouting to Line 2. Therefore, the optimal path is $\alpha_{RM}^{\ast}=\langle q^{1},q^{2},q^{3},q^{4},q^{21},q^{11},q^{12},q^{13},q^{8} \rangle$ with objective function value $V^{\ast}=18$, according to Equation~\ref{eq-costfunction-modified}. Note that as re-routing occurs, the PCM being used is no longer that of Line 3, rather it is that of the optimal path with $\kappa^{\ast}=2$. Line 3 is designed to have 3 branch states each having an out-degree centrality of 2 and therefore is the most PCM optimal (see Statement~\ref{theorem-centrality2}). Since the redundant paths are dispersed through Line 3, the chance of failure is minimum among all paths with a length of 8 and the sum of out-degree centrality of 10. Also, according to Lemma~ \ref{corollay-optimalriskmeasure}, the sub-path traversed prior to failure, i.e.  $\hat{\alpha}=\langle q^{1},q^{2},q^{3},q^{4}\rangle$ is still the most optimal sub-path.

\par
Scenario 2 is similar to Scenario 1 with the difference that the failure in Line 1 occurs before the product passes the branch state, in which case the CB PTA-MPC switches to Line 2 and the optimization is complete with the optimal path being $\alpha_{CB}^{\ast}=\langle q^{1},q^{14},q^{15},q^{24},q^{11},q^{12},q^{13},q^{8} \rangle$ and objective function value of $V^{\ast}=16$. Comparing the solution of CB PTA-MPC to that of PCM depicts that the former's solution is optimum in terms of cost and risk overall. However, as mentioned in Scenario 1, CB PTA-MPC is more prone to failure compared to PCM PTA-MPC. Therefore, although in certain cases CB PTA-MPC can have an optimum solution overall, it is relatively more prone to failure compared to the PCM PTA-MPC.
\par
The common aspect among all strategies is that the control algorithm is \textit{aware} of the current workstation failures. However, based on the risk-averse route scheduling, the outcome is different. PTA-MPC scheme only persuades a minimum cost, whereas CB PTA-MPC quantifies the risk based on the sum of out-degree centrality measures. Finally, PCM PTA-MPC quantifies the average number of branch-states per path length improving the risk-averseness and showing relatively better performance in face of state failure. Despite this, under certain circumstances which to be discussed in future work, the results of the proposed framework might not be the most cost-optimal. The above example illustrated a case in which both CB PTA-MPC and PCM PTA-MPC reached the objective with CB PTA-MPC having a lower objective function value and therefore, performing better. As discussed, depending on which state failing and when, these peculiarities might occur, however, in majority of cases PCM PTA-MPC shows more flexibility to state failure compared to CB PTA-MPC. 
%\par
%The cost objective function value depends on the states that comprise a path. As occurrence of state failure is uncertain, the relative optimality of risk-averse PTA-MPC compared to PTA-MPC is also uncertain. While PTA-MPC is more optimal on certain occasions, the risk-averse PTA-MPC provides a more flexible solution which is robust to state failure. 

\section{Conclusion and Future Work}
\label{sec-conclusion}

In this work, we have proposed a risk-averse PTA-MPC that is based on the automaton feature and performs relatively safer compared to centrality-based risk-averse PTA-MPC. As part of this framework, we defined Path Commitment Measure in terms of Committed Sub-Path. We further defined the significance factor to relate cost and the path commitment measure that allows the reduction of a  multi-objective constrained optimization into cost-dependent single-objective optimization. Finally, we developed risk-averse PTA-MPC algorithm which returns a robust solution considering both risk and cost. The proposed framework was shown to allow the system to reach the goal state in the face of state failure comparatively better compared to centrality-based risk-averse algorithms when applied in a manufacturing system. It was also shown that under certain conditions, centrality-based risk averse algorithm can have a better result depending on the nature and time of failure. It was shown that the commitment measure risk averse PTA-MPC is more robust in face of failure.
This work can further be extended to other application areas, such as mass-production, supply-chain and can be tested on real systems using real computer solvers. In addition, a hybrid controller consisting of CB PTA-MPC and PCM PTA-MPC can be developed to leverage the benefits of both controllers based on the PTA architecture and the manufacturing system environment. 
\vfill
\null
%\printbibliography

\end{document}